\documentclass[letterpaper]{article} 
\usepackage{aaai24}  
\usepackage{times}  
\usepackage{helvet}  
\usepackage{courier}  
\usepackage[hyphens]{url}  
\usepackage{graphicx} 
\usepackage{natbib}  
\usepackage{caption} 
\usepackage{newfloat}
\usepackage{listings}
\DeclareCaptionStyle{ruled}{labelfont=normalfont,labelsep=colon,strut=off} 
\usepackage{soul}
\usepackage{url}
\usepackage[utf8]{inputenc}
\usepackage{amsmath}
\usepackage{amsthm}
\usepackage{booktabs}

\usepackage{xcolor}
\usepackage{mathrsfs}
\usepackage{booktabs}
\usepackage[ruled]{algorithm2e}
\usepackage{multirow}
\usepackage{array}

\SetKwInput{KwInput}{Input}
\SetKwInput{KwOutput}{Output}
\DeclareMathOperator{\OPT}{OPT}
\DeclareMathOperator{\argmin}{\arg\min}

\newtheorem{theorem}{Theorem}
\newtheorem{corollary}{Corollary}
\newtheorem{definition}{Definition}
\newtheorem{lemma}{Lemma}

\usepackage[utf8]{inputenc} 
\usepackage[T1]{fontenc}    
\usepackage{url}            
\usepackage{booktabs}       
\usepackage{amsfonts}       
\usepackage{nicefrac}       
\usepackage{microtype}      
\usepackage{xcolor}         

\title{Cost Minimization for Equilibrium Transition}
\author {
    Haoqiang Huang\textsuperscript{\rm 1},
    Zihe Wang\textsuperscript{\rm 2}\thanks{Zihe Wang is the corresponding author.},
    Zhide Wei\textsuperscript{\rm 3},
    Jie Zhang\textsuperscript{\rm 4}
}
\affiliations {
    \textsuperscript{\rm 1}Hong Kong University of Science and Technology\\
    \textsuperscript{\rm 2}Renmin University of China\\
    \textsuperscript{\rm 3}Peking University\\
    \textsuperscript{\rm 4}University of Bath\\
    hhuangbm@connect.ust.hk, 
    wang.zihe@ruc.edu.cn,
    zhidewei@pku.edu.cn,
    jz2558@bath.ac.uk
}

\usepackage{bibentry}

\begin{document}

\maketitle

\begin{abstract}
  In this paper, we delve into the problem of using monetary incentives to encourage players to shift from an initial Nash equilibrium to a more favorable one within a game. Our main focus revolves around computing the minimum reward required to facilitate this equilibrium transition. The game involves a single row player who possesses $m$ strategies and $k$ column players, each endowed with $n$ strategies. Our findings reveal that determining whether the minimum reward is zero is NP-complete, and computing the minimum reward becomes APX-hard. Nonetheless, we bring some positive news, as this problem can be efficiently handled if either $k$ or $n$ is a fixed constant. Furthermore, we have devised an approximation algorithm with an additive error that runs in polynomial time. Lastly, we explore a specific case wherein the utility functions exhibit single-peaked characteristics, and we successfully demonstrate that the optimal reward can be computed in polynomial time.
\end{abstract}

\section{Introduction}\label{Sec: Intro}
	Equilibrium analysis has remained a central focus in game theory research for several decades. The exploration of the fundamental Nash equilibrium and its various refinements has been the subject of extensive study. It is widely recognized that certain equilibria hold greater desirability than others, prompting investigations into the Price-of-Anarchy \cite{DBLP:journals/csr/KoutsoupiasP09} and Price-of-Stability \cite{DBLP:journals/siamcomp/AnshelevichDKTWR08}. In the dynamics of the game, players are driven to best respond to their counterparts' strategies \cite{hopkins1999note,leslie2020best}. Nevertheless, this approach may lead to suboptimal equilibria, as players overlook the potential benefits that could arise from deviating from the best response dynamics in response to successive changes made by other players.
	
	In our research, we explore a scenario where a mediator aims to facilitate the transition from an initial equilibrium to a more favorable target equilibrium. The mediator achieves this by subsidizing the players to influence their behavior, encouraging them to progress towards the desirable equilibrium in gradual steps. Throughout this process, the mediator's primary objective is to minimize the overall cost involved. This study holds significant real-world applications. For instance, many governments worldwide are actively pursuing initiatives to boost the adoption of electric vehicles as part of their net-zero plans. Given that cars and vans contribute nearly a fifth of total emissions, expediting the transition to electric vehicles is crucial to accomplishing their environmental goals. One can liken the initial equilibrium to the prevailing usage of petrol and diesel cars. To drive the shift towards electric vehicles, the government may implement tax benefits and provide funds for charger installations, thus incentivizing drivers to make the switch. The more desirable target equilibrium, in this case, would involve completely phasing out the sale of new petrol and diesel cars.

   In this intriguing game, we have a row player and $k$ column players, each with specific payoffs represented by matrices $R$ and $C$ respectively. The game commences from an initial equilibrium and proceeds in rounds. As a mediator, the ultimate objective is to lead the players towards a desirable Nash equilibrium by offering rewards in each round. Even in the case of just one row player, the theoretical implications are already noteworthy. In practical terms, this single row player scenario effectively captures numerous monopolistic markets and the regulatory behavior of governments in markets. However, as we will demonstrate, the problem at hand becomes intractable in various setups. Thus, achieving positive outcomes in more general settings necessitates the application of additional constraints or the consideration of special cases.
    

	\subsection{Our contribution}
	In this paper, we tackle the challenge of designing algorithms for computing the minimum reward needed to foster equilibrium transitions. Given a bimatrix game with payoff matrices $R_{m \times n}$ and $C_{m \times n}$, assuming that there are $k$  column players playing against a row player, we show the following results: 
	\begin{itemize}
		\item Determining whether the minimum reward is zero is NP-complete, and computing the minimum reward, in general, is APX-hard.
		\item However, computing the optimal reward scheme is slicewise polynomial with respect to $k$ and $n$, respectively.
		\item We design an approximation algorithm for this problem that runs in polynomial time. The additive approximation error is linear in the number of the row player's choices and the largest number of matrices $R$ and $C$.
		\item Last, we consider a special case where the utility functions are single-peaked and show that the optimal reward can be computed in polynomial time. 
	\end{itemize}
	We show the hardness results by reductions from the EXACT COVER problem and a variant of the Knapsack problem. In the approximation algorithm, we construct a complete directed graph in which the weight of an edge corresponds to the solution of an integer linear programming (ILP). By approximating the solution of the ILP, the problem of finding the optimal transformation path from an initial equilibrium to a target equilibrium boils down to finding the shortest paths between vertices of the graph. 
	
	
	\subsection{Related Work}
	 \citet{DBLP:journals/jair/MondererT04} consider the implementation of desirable outcomes by a reliable party who cannot modify game rules but can make non-negative payments to the players. They term this $k$-implementation problem an intermediate approach between algorithmic game theory \cite{DBLP:conf/www/Papadimitriou11} and mechanism design \cite{DBLP:conf/stoc/NisanR99}. They provide characterizations of $k$-implementation for the implementation of singletons in games with complete information and investigate several settings under which the problem is polynomial-time solvable or intractable. 
	\citet{DBLP:conf/atal/DengTZ16} follow up the study of $k$-implementation and prove that the problem is NP-complete for general games with respect to dominance by pure strategies. Furthermore, the authors study a variation of the $k$-implementation problem by characterizing its hardness and developing computationally efficient algorithms for supermodular games.
    Unlike $k$-implementation, which considers a single-round bi-matrix game, we aim to motivate a group of players to move to a target equilibrium step by step and minimize the total cost.
	\citet{DBLP:conf/aaai/DengC17} consider a disarmament game, in which players successively commit not to play certain strategies and thereby iteratively reduce their strategy spaces. 
	Later on, in \cite{DBLP:conf/aaai/DengC18}, instead of removing a strategy in a game, they consider removing a resource which leads to ruling out all the strategies in which that resource is used simultaneously. They prove NP-completeness of several formulations of the problem of achieving desirable outcomes via disarmament.
	
	The idea of allowing a mediator to influence the players' behavior and hence the outcome of a system has been widely studied. For example,  \citet{DBLP:conf/ijcai/RozenfeldT07} focus on the use of routing mediators in order to reach a correlated strong equilibrium. They show that a natural class of routing mediators allows the implementation of fair and efficient outcomes as a correlated super-strong equilibrium in a very wide class of games. 
	\citet{DBLP:journals/ai/MondererT09} propose to use mediators in order to enrich the set of situations where one can obtain stability against deviations by coalitions, in light of the understanding that strong equilibrium rarely exists. 
	\citet{DBLP:journals/algorithmica/AugustineCFK15} address the question of whether a network designer can enforce particular equilibria or guarantee that efficient designs are consistent with users' selfishness by appropriately subsidizing some of the network links. They formulate this question as one of the optimization problems and present positive and negative results.
	\citet{DBLP:conf/cocoa/EidenbenzOSW07} consider the problem of a mechanism designer seeking to influence the outcome of a strategic game based on her creditability. 

      Studying the best response dynamics of players constitutes a fundamental aspect of game theory research. In a recent study, \citet{DBLP:journals/mor/AmietCSZ21} explore games where the payoffs are drawn at random and demonstrate that a best-response dynamics approach leads to a pure Nash equilibrium with a high probability as the number of players increases. \citet{DBLP:conf/sagt/FeldmanST17} delve into congestion games, where they investigate the inefficiency of various deviator rules. They find that the best response dynamics consistently converges to a pure Nash equilibrium in such games. \citet{heinrich2022bestresponse} analyze the performance of the best-response dynamic across all normal-form games using a random games approach. They show that the best-response dynamic converges to a pure Nash equilibrium in a vanishingly small fraction of all large games when players take turns according to a fixed cyclic order. By contrast, when the playing sequence is random, the dynamic converges to a pure Nash equilibrium if one exists in almost all large games.

\section{Preliminary}\label{Sec: Pre}
	We consider a game in which there is a row player and $k$ column players. The row player and each of the $k$ column players constitute a bimatrix game. Let $\mathscr{R}$ be the strategy set of the row player and $\mathscr{C}$ be the strategy set of a column player. Denote $\mathscr{C}^{k}$ the set of all possible strategy profiles of $k$ column players. Let $r^{i}$ and $c^{j}$ be the row player's $i$-th strategy and a column player's $j$-th strategy, respectively, where $i=1,\dots,m$ and $j=1,\dots,n$. Throughout the paper, the players only adopt pure strategies. Let $R_{m\times n}$ and $C_{m\times n}$ be the payoff matrices of the row player and the column players, respectively.  Denote $(r(t), c_1(t),\dots, c_k(t))$ the strategy profile of all players at time step $t$, where $r(t) \in \mathscr{R}$ and $c_i(t)\in \mathscr{C}, i=1,\dots,k$. For ease of notation, we denote $\mathcal{C}(t) = (c_1(t),\dots,c_k(t)) \in \mathscr{C}^{k}$. This way, the row player's payoff at time $t$ is $\sum_{1\le i\le k}R(r(t), c_i(t))$, i.e., the sum of its payoff in the $k$ bimatrix games. The payoff of the column player $i$ is $C(r(t), c_i(t))$, i.e., its payoff in the bimatrix game that it plays against the row player.
	
	In the context of equilibrium transition, the game starts with an equilibrium. That is, the row player and the column players are at an equilibrium in each of the $k$ bimatrix games. Assume that, from the mediator's perspective, there exists another more desirable equilibrium. The mediator is interested in designing an optimal reward scheme that motivates the players to move to the more desirable equilibrium over multiple rounds.
	
	Specifically, consider a strategy profile $(r(t), \mathcal{C}(t))$ in round $t$, without any additional reward, the row player will best respond to the column players' strategy profile $\mathcal{C}(t)$ in the next round. Therefore, to incentivize the row player playing a specific strategy $r(t+1) \in \mathscr{R}$ in the next round, we need to provide a reward of, denoted by $T_{\mathcal{C}(t)}(r(t+1))$,
	\begin{align*}
		\max_{r^i\in \mathscr{R}} \sum_{1\le j\le k}R(r^i,c_j(t))-\sum_{1\le j\le k}R(r(t+1),c_j(t)),
	\end{align*}
	where the first term is the maximum payoff that the row player can get by taking the best response strategy against column players' strategies $\mathcal{C}(t)$, and the second term is the row player's payoff when it takes the strategy $r(t+1) \in \mathscr{R}$. 
	
	

	Similarly, given the row player's strategy $r(t)$ in round $t$, to incentivize column players taking strategy profile $\mathcal{C}(t+1)$ in round $t+1$, the total reward needed is
	$$T_{r(t)}(\mathcal{C}(t+1)) := k\cdot \max_{c^i\in \mathscr{C}} C(r(t),c^i)-\sum_{1\le j\le k}C(r(t),c_j(t+1)).$$
	
	
	Given the initial equilibrium $(r(1),\mathcal{C}(1))$ and the target equilibrium $(r^{*}, \mathcal{C}^{*})$, a reward scheme that incentivizes the players moving from the initial equilibrium to the target equilibrium consists of a \emph{transformation path} 
	$(r(1),\mathcal{C}(1)) \rightarrow (r(2), \mathcal{C}(2)) \rightarrow \dots \rightarrow (r^{*}, \mathcal{C}^{*})$. The total cost of this reward scheme is $T = \sum_t \{ T_{r(t)}(\mathcal{C}(t+1)) + T_{\mathcal{C}(t)}(r(t+1)) \}$.
	
We define the optimization problem OPT TRANSITION $(k,m,n)$ as follows.
	
	\noindent {\bf Problem 1:} OPT TRANSITION $(k,m,n)$.\\
	{\bf Input:} Payoff matrices $R_{m\times n}$ and $C_{m\times n}$. The initial equilibrium $(r(1),\mathcal{C}(1))$ and the target equilibrium $(r^{*}, \mathcal{C}^{*})$, \\
	{\bf Output:} A transformation path from strategy profile $(r(1),\mathcal{C}(1))$ to $(r^{*}, \mathcal{C}^{*})$ such that the total cost 
	$T = \sum_t \{ T_{r(t)}(\mathcal{C}(t+1)) + T_{\mathcal{C}(t)}(r(t+1)) \}$ is minimized.
	
	We also consider the following decision problem, called TRANSITION $(k,m,n,T)$.
	
	\noindent {\bf Problem 2:} TRANSITION $(k,m,n,T)$.\\
	{\bf Input:} Payoff matrices $R_{m\times n}$ and $C_{m\times n}$. The initial equilibrium $(r(1),\mathcal{C}(1))$ and the target equilibrium $(r^{*}, \mathcal{C}^{*})$, \\
	{\bf Output:} YES, if a transformation path from $(r(1),\mathcal{C}(1))$ to $(r^{*}, \mathcal{C}^{*})$ with the total cost no larger than
	$T$ exists, and NO otherwise.
	

	\section{Complexity Results}\label{Sec: Hardness}
	This section investigates the complexity of the above two problems. We show that OPT TRANSITION $(k,m,n)$ is APX-hard, which discourages us from designing efficient algorithms that can find a solution within some fixed multiplicative factor of the optimal cost $T$. In addition, even for the case that the row player has only two strategies, the decision problem TRANSITION $(k,2,n,T)$ is NP-complete. However, OPT TRANSITION $(k,m,n)$ becomes polynomial-time solvable when either $k$ or $n$ is a fixed constant. 
	
	\subsection{General Values of $k,m$, and $n$}
	We show that the decision problem TRANSITION $(k,m,n,T)$ is NP-complete when $T=0$. 
	\begin{theorem}
		TRANSITION $(k,m,n,0)$ is NP-complete.
	\end{theorem}
	\begin{proof}
		Given a transition path from the initial equilibrium to the target equilibrium, it is easy to verify whether it is a valid path and its cost is 0. For NP-hardness, we reduce from the EXACT COVER problem, which is known to be NP-complete \cite{DBLP:conf/coco/Karp72}, to the decision problem TRANSITION $(k,m,n,0)$. Recall the EXACT COVER problem:
		
		\noindent {\bf Problem 3:} EXACT COVER $(s,w)$.\\
		{\bf Input:} A finite set $Z=\{1,2,\dots,3s\}$ and a collection $X = \{X_1,X_2,\dots,X_w\}$ of 3-element subsets of the set $Z$, where $s \le w$. \\
		{\bf Output:} YES, if there exists a collection $\{X_{i_1},X_{i_2},\dots, X_{i_s} \} \subset X$ such that their union is $Z$, and NO otherwise.
		
		Given an instance of EXACT COVER, we construct an instance of TRANSITION $(k,m,n,0)$ as follows:
		
		Let $k=s$. That is, there are $s$ column players. Let $m=3s+2$ and $n=w+2$. We design the column players' payoff matrix $C_{(3s + 2) \times (w + 2)}$ as below. Except for the last row and the last column, all elements in matrix $C$ are $1$. The intersection of the last row and last column is set to be $1$ as well. Lastly, all the remaining elements are $0$. 
		
		{
			\begin{table}[ht]
				\centering
	        \setlength{\extrarowheight}{2pt}
          \hskip-0.75cm
				\begin{tabular}{cc|c|c|c|c|}
					& \multicolumn{1}{c}{} & \multicolumn{3}{c}{Column Players' Payoff Matrix $C$}\\
					& \multicolumn{1}{c}{} & \multicolumn{1}{c}{$1$}  & \multicolumn{1}{c}{$\cdots$} & \multicolumn{1}{c}{$w + 1$} & \multicolumn{1}{c}{$w + 2$} \\\cline{3-6}
					\multirow{3}*{}  & $1$ & $1$ & $\cdots$ & $1$ & $0$\\\cline{3-6}
					& $\vdots$ & $\vdots$ & $\vdots$ & $\vdots$ & $\vdots$\\\cline{3-6}
					& $3s + 1$ & $1$ & $\cdots$ & $1$ & $0$\\\cline{3-6}
					& $3s + 2$ & $0$ & $\cdots$ & $0$ & $1$\\\cline{3-6}
				\end{tabular}
			\end{table}
		}
		Let $v_1,v_2,\dots,v_w$ denote the characteristic vectors of $ X_1,X_2,\dots,X_w$. That is, for a vector $v_i=(v_{i,j})_{j=1,\dots,3s}$, its elements $v_{i,j}=1$ if $j\in X_{i}$ and $v_{i,j}=0$ if $j\notin X_{i}$, $i=1,\dots,w$. We construct the row player's payoff matrix $R_{(3s + 2) \times (w + 2)}$ as follows. The entries in the first row are $0$ except the last one being $-1$. The entries in the first column are $0$ except the last one being $-s$. Then, for the last row, $R_{(3s+2),j}=1/s, j=2,\dots,w+1$ and $R_{(3s+2),(w+2)}=1$. For the last column, $R_{i,(w+2)}=0, i=2,\dots,3s+1$. The other entries of matrix $R$ are filled by elements $v_{i,j}$ as shown below.

		{  
			\begin{table}[ht]

                \centering
				\setlength{\extrarowheight}{2pt}
  \hskip-0.75cm
				\begin{tabular}{cc|c|c|c|c|c|}
					& \multicolumn{1}{c}{} & \multicolumn{3}{c}{Row Player's Payoff Matrix $R$}\\
					& \multicolumn{1}{c}{} & \multicolumn{1}{c}{$1$}  & \multicolumn{1}{c}{$2$} & \multicolumn{1}{c}{$\cdots$} & \multicolumn{1}{c}{$w + 1$} & \multicolumn{1}{c}{$w + 2$} \\\cline{3-7}
					\multirow{3}*{}  & $1$ & $0$ & $0$ & $\cdots$ & $0$ & $-1$\\\cline{3-7}
					& $2$ & $0$ & $v_{1,1}$ & $\cdots$ & $v_{w,1}$ & $0$\\\cline{3-7}
					& $\vdots$ & $\vdots$ & $\vdots$ & $\vdots$ & $\vdots$ & $\vdots$\\\cline{3-7}
					& $3s + 1$ & $0$ & $v_{1,3s}$ & $\cdots$ & $v_{w,3s}$ & $0$\\\cline{3-7}
					& $3s + 2$ & $-s$ & $1/s$ & $\cdots$ & $1/s$ & $1$\\\cline{3-7}
				\end{tabular}
			\end{table}
 
		}

		We note that there are at least two pure Nash equilibria in the $(k+1)$-player game. They are $(r^1, c^1,\dots,c^1)$ and $(r^{3s+2},c^{w+2},\dots,c^{w+2})$; namely, all players choose their first strategy and all players choose their last strategy, respectively. In particular, let the initial Nash equilibrium $(r(1),\mathcal{C}(1))$ be $(r^1, c^1,\dots,c^1)$ and the target equilibrium $(r^{*}, \mathcal{C}^{*})$ be $(r^{3s+2},c^{w+2},\dots,c^{w+2})$. Till now, we have constructed an instance of TRANSITION $(k=s,m=3s+2,n=w+2,0)$. 
		
		\noindent {\bf Reduction correctness.} Given that there is a solution to TRANSITION $(s,3s+2,w+2,0)$, we can compute a solution to EXACT COVER $(s,w)$. To this end, we disclose the features of a zero-cost transformation path of TRANSITION $(s,3s+2,w+2,0)$.
		
		First, we observe from payoff matrix $C$, that a column player will not choose to play its last strategy $c^{w+2}$ unless the row player has chosen its last strategy $r^{3s+2}$. Based on this, we observe from payoff matrix $R$, that the row player will not choose to play its last strategy $r^{3s+2}$ as long as there is a single column player playing its first strategy $c^{1}$. This is because, in that case, the row player's payoff is at most $-s + \frac{s-1}{s}$ which is less than 0. In view of these observations, we conclude that the $s$ column players must be playing some strategies within the set $\{c^{2},\dots,c^{w+1}\}$ on a zero-cost transformation path from the initial equilibrium $(r^1, c^1,\dots,c^1)$ to the target equilibrium $(r^{3s+2},c^{w+2},\dots,c^{w+2})$. 
		
		Second, for the second-last step on the transition path, we notice that these $s$ column players must choose $s$ different strategies amongst the set $\{c^{2},\dots,c^{w+1}\}$. Otherwise, suppose that two column players are choosing the same strategy, for example, $c^{2}$. Note that the characteristic vector $v_1 = (v_{1,j})_{j=1,\dots,3s}$ has exactly three elements whose value are 1. Without loss of generality, assume $v_{1,1}=1$. Then, by playing strategy $r^2$, the row player's payoff is 2, which is greater than its payoff 1 by playing strategy $r^{3s+2}$. Therefore, it contradicts the existence of a zero-cost transformation path on which the row player will transit to playing the last strategy.
		
		Last, we notice that in these $s$ columns of payoff matrix $R$ that correspond to the $s$ different strategies played by these column players, there should be only one 1-element in each row. Otherwise, by playing the strategy that corresponds to a row that has multiple 1-elements, the row player has a higher payoff and will not transit to playing the last strategy without a positive reward. The same contradiction occurs. So, the row player is indifferent between playing the last strategy $r^{3s+2}$ and any one of the strategies in $\{r^{2},\dots,r^{3s+1}\}$, as its payoff is 1 in either case. Together with the fact that there are $3s$ rows (namely, row 2 to row $3s+1$) and each of these $s$ columns has three 1-elements, we conclude that these $s$ columns correspond to $s$ characteristic vectors of 3-elements sets such that their union is the set $Z=\{1,2,\dots,3s\}$.
		
		To conclude, given a zero-cost transition path from the initial equilibrium to the target equilibrium, we can identify a set of values $v_{i,j}$ such that their corresponding 3-elements subsets $\{X_{i_1},X_{i_2},\dots,X_{i_s}\}$ is a solution to the SET COVER problem. Also, if the SET COVER problem has a solution $\{X_{i_1},\dots,X_{i_s}\}$ such that their union is $Z$, then we can construct the transition path as 
		$(r^1,c^1,\dots,c^1)\rightarrow (r^1,C^{i_1},C^{i_2},\dots,C^{i_s})\rightarrow (r^{3s+2},C^{i_1},C^{i_2},\dots,C^{i_s})\rightarrow (r^{3s+2},C^{w+2},C^{w+2},\dots,C^{w+2})$. It is easy to verify that the cost of this transition path is zero. 
		
	Hence, TRANSITION $(k,m,n,0)$ is NP-complete.
	\end{proof}
	
	This result immediately implies that the problem OPT TRANSITION $(k,m,n)$ is APX-hard.
	
	\begin{corollary}
		Computing the optimal transformation path is APX-Hard and no multiplicative approximation is possible. That is, constructing a transformation path of cost $\alpha \cdot \OPT$ is NP-Hard for any $\alpha > 1$, where $\OPT$ is the cost of the optimal transformation path.
	\end{corollary}

	\subsection{When one of the variables, $k$, $m$, $n$, is a constant}
    Since we have shown that TRANSITION $(k,m,n,T)$ is APX-hard, we then consider special cases when one of the variables, $k$, $m$, $n$, is a constant. 
    
    {\bf When $m=2$.} The following theorem shows that the problem TRANSITION $(k,m,n,T)$ is hard to solve even if the row player has only two strategies. 
	
	\begin{theorem}\label{TRANSITION}
		TRANSITION $(k,2,n,T)$ is NP-complete.
	\end{theorem}

\begin{proof}
	Membership in NP follows trivially from the problem definition. 
	We show the NP-hardness by reducing from a variant of the Knapsack problem, that has an exact number of item constraints, to TRANSITION $(k,2,n,T)$. For completeness, we include proof that this variant of the Knapsack problem is NP-hard in Lemma \ref{KnapsackNPHard}.
	
	\noindent {\bf Problem 4:} Knapsack $(m,k,W,V)$.\\
	{\bf Input:} A set of $m$ items, each with a weight  $w_{i}$ and a value $v_{i}$, where $i=1,\dots,m$. A maximum weight capacity $W$, where $W \ge w_i \ge 0, \forall \,\ i$. A target value $V$, where $V \ge v_i \ge 0,  \forall \,\ i$. A constant integer $k$. \\
	{\bf Output:} YES, if there exists a multiset $I =\{i_1,\dots,i_k\} , i_j \in [m]$, such that
	$\sum_{j = 1}^k v_{i_j} \geq V$ and $\sum_{j = 1}^k w_{i_j} \leq W $, and NO otherwise. 
	
	Note that, in this variant of the Knapsack problem, a valid solution must contain exactly $k$ items in the knapsack.
	
	Given an instance of Knapsack $(m,k,W,V)$, we construct an instance of TRANSITION $(k + 1,2,m + 3,4 \epsilon k V)$ in which there are $k+1$ column players, each has $m+3$ strategies, and $\epsilon$ is sufficiently small such that $\epsilon \ll \frac{w_i}{kV}, \forall i$. Specifically, the row player and the column players' payoff matrices are as follows.
	
	
	{
		\begin{table}[h]
			\centering
			\setlength{\extrarowheight}{2pt}
			\begin{tabular}{cc|c|c|c|c|c|c|}
				& \multicolumn{1}{c}{} & \multicolumn{5}{c}{Column Players' Payoff Matrix $C$}\\
				& \multicolumn{1}{c}{} & \multicolumn{1}{c}{$0$}  & \multicolumn{1}{c}{$1$} & \multicolumn{1}{c}{$\cdots$} & \multicolumn{1}{c}{$m$} & \multicolumn{1}{c}{$m + 1$} & \multicolumn{1}{c}{$m + 2$} \\\cline{3-8}
				\multirow{3}*{}  & $1$ & $ \epsilon V$ & $\epsilon v_1$ & $\cdots$ & $\epsilon v_m$ & $-k  \epsilon V$ & $-\infty$\\\cline{3-8}
				& $2$ & $0$ & $0$ & $\cdots$ & $0$ & $0$ & $1$ \\\cline{3-8}
			\end{tabular}
		\end{table}
	}
	
	{
		\begin{table}[h]
			\centering
			\setlength{\extrarowheight}{2pt}
			\begin{tabular}{cc|c|c|c|c|c|c|}
				& \multicolumn{1}{c}{} & \multicolumn{5}{c}{Row Player's Payoff Matrix $R$}\\
				& \multicolumn{1}{c}{} & \multicolumn{1}{c}{$0$}  & \multicolumn{1}{c}{$1$} & \multicolumn{1}{c}{$\cdots$} & \multicolumn{1}{c}{$m$} & \multicolumn{1}{c}{$m + 1$} & \multicolumn{1}{c}{$m + 2$} \\\cline{3-8}
				\multirow{3}*{}  & $1$ & $0$ & $0$ & $\cdots$ & $0$ & $0$ & $-\infty$\\\cline{3-8}
				& $2$ & $-\infty$ & $-w_1$ & $\cdots$ & $-w_m$ & $W$ & $2W$ \\\cline{3-8}
			\end{tabular}
		\end{table}
	}
	
	We note that there are at least two pure Nash equilibria in this $(k+2)$-player game. They are $(r^1, c^0,\dots,c^0)$ and $(r^{2},c^{m+2},\dots,c^{m+2})$; namely, all players choose their first strategy and all players choose their last strategy, respectively. In particular, let the initial Nash equilibrium $(r(1),\mathcal{C}(1))$ be $(r^1, c^0,\dots,c^0)$ and the target equilibrium $(r^{*}, \mathcal{C}^{*})$ be $(r^{2},c^{m+2},\dots,c^{m+2})$. Till now, we have constructed an instance of the problem TRANSITION $(k + 1,2,m + 3,4 \epsilon k V)$.
	
	\noindent {\bf Reduction correctness.} To show the correctness of the reduction, we disclose the features of a solution to TRANSITION $(k + 1,2,m + 3,4 \epsilon k V)$ that admits a transformation path whose cost is no more than $4 \epsilon k V$.
	
	First, note that a column player's payoff is negative infinity by playing $c^{m+2}$ when the row player is playing its first strategy $r^1$. So, none of the column players will switch to their last strategy $c^{m+2}$ with bounded reward, unless the row player has switched to its second strategy $r^2$. In other words, the last step of the transformation path must be $(r^2,\mathcal{C}(t)) \rightarrow (r^{2},c^{m+2},\dots,c^{m+2})$, where $\mathcal{C}(t)$ is the strategy profile of the column players in the second last round. Denote $I_{\mathcal{C}(t)}$ the set of indices of the strategy profile $\mathcal{C}(t)$. Then $I_{\mathcal{C}(t)}$ is a multiset of size $k + 1$ whose distinct elements belong to $\{0,1,\dots,m+1\}$. We also note that the row player's payoff is negative infinity when it plays $r^2$ even if only one of the column players is playing its first strategy $c^0$. So, $I_{\mathcal{C}(t)}$ is a multiset of size $k + 1$ whose distinct elements belong to $\{1,\dots,m+1\}$. Therefore, in this last round $(r^2,\mathcal{C}(t)) \rightarrow (r^{2},c^{m+2},\dots,c^{m+2})$, no positive reward is needed for the players as their payoffs are all increased. 
	
	Second, denote the round before $(r^2,\mathcal{C}(t))$ by $(r^1,\mathcal{C}(t-1))$, which is evolved from the initial equilibrium $(r^1, c^0,\dots,c^0)$. That is, $(r^1, c^0,\dots,c^0) \rightarrow (r^1,\mathcal{C}(t-1)) \rightarrow (r^2,\mathcal{C}(t))$, in which $\mathcal{C}(t-1)$ is similar to $\mathcal{C}(t)$ in the sense that it corresponds to a multiset of size $k + 1$, consisting of the index of the strategies chosen by column players. We note that in either round, no positive reward is needed for the row player. This is obvious for the first round since the row player's payoff remains $0$. In the second-last round $(r^1,\mathcal{C}(t-1)) \rightarrow (r^2,\mathcal{C}(t))$, we can further conclude that there exists at least one column player who plays strategy $c^{m+1}$, since otherwise, it costs at least a $w_i$ reward to incentivize the row player choosing strategy $r^2$, which contradicts with the scale of the total cost $T=4 \epsilon k V$ (by the choice of $\epsilon \ll \frac{w_i}{kV}, \forall i$). Denote the columns in $C(t)$ as $i_1,\dots,i_k$. Because the row player's payoff after shifting the strategy is not smaller than its payoff before that, we have the inequality $\sum_{j = 1}^k w_{i_j} \leq W $. In addition, we note that the costs needed to incentivize the column players switching from $(r^1, c^0,\dots,c^0)$ to $(r^1,\mathcal{C}(t-1))$ and from $(r^1,\mathcal{C}(t-1))$ to $(r^2,\mathcal{C}(t))$ are the same, owing to the fact that the row player switched from the same strategy $r^1$ and the structure of $\mathcal{C}(t-1)$ and $\mathcal{C}(t)$. Hence, in either step, the cost needed to incentivize the column players is $2\epsilon k V$.
	
	Third, to ensure that the cost needed to incentivize the column players switching from $(r^1, c^0,\dots,c^0)$ to $(r^1,\mathcal{C}(t-1))$ is at most $2\epsilon k V$, there must be only one column player who plays the $(m+1)$-th strategy $c^{m+1}$. This is because, only in this case, the cost is 
	\begin{align*}
		& \epsilon \Big(\sum_{{j} = 1}^k (V - v_{i_j}) + (V + kV) \Big) \\
		= &  \epsilon \Big( 2kV + (V - \sum_{j = 1}^k v_{i_j}) \Big) \\
		\leq &  2k \epsilon V ,
	\end{align*}
	where the last inequality is owing to $\sum_{j = 1}^k v_{i_j} \ge V$.
	
	Therefore, a solution to TRANSITION $(k + 1,2,m + 3,4 k V \epsilon)$ implies a solution to  Knapsack $(m,k,W,V)$. Also, if the problem $Knapsack(m,k,W,V)$ has a solution $I=\{i_1,\dots,i_k\}$, then the transition path $(r^1,C^0,\dots,C^0)\rightarrow 
	(r^1,C^{i_1},C^{i_2},\dots,C^{i_k},C^{m+1})\rightarrow 
	(r^2,C^{i_1},C^{i_2},\dots,C^{i_k},C^{m+1})\rightarrow
	(r^2,C^{m+2},\dots,C^{m+2})$ has a cost of zero.
	
	Therefore, $TRANSITION(k,2,n,T)$ is NP-complete.
	
\end{proof}

\begin{lemma}\label{KnapsackNPHard}
	Problem 4, Knapsack $(m,k,W,V)$, is NP-hard.
\end{lemma}

\begin{proof}
	For the completeness of the paper, we include a proof that the variant of Knapsack appearing in Theorem \ref{TRANSITION} is NP-hard. We reduce the classical Knapsack problem, which is known to be NP-complete, to Knapsack $(m,k,W,V)$. To this end, we adapt the weight and value of the items, as well as the target value and weight in the classical Knapsack problem, so that a valid solution must contain exact $k$ items. 
	
	For each item $i$ in the classical Knapsack problem, $i=1,\dots,m$, add $m(V+1)$ to its value $v_i$ so that it becomes $v'_i = v_i + m(V+1)$, and add $W+1$ to its weight $w_i$ so that it becomes $w'_i = w_i + (W+1)$. In addition, let $V' = V + k m (V+1)$ and $W' = W + k (W+1)$ be the target value and target weight in Knapsack $(m,k,W,V)$. Since $V \ge v_i, \forall i$, it is clear that in order to have $\sum_i v'_i \ge V'$, the solution needs to include at least $k$ items in the Knapsack. Similarly, since $W \ge w_i, \forall i$, it is clear that in order to have $\sum_i w'_i \le W'$, the solution needs to include at most $k$ items in the Knapsack. To conclude, the solution to Knapsack $(m,k,W,V)$ must have exact $k$ items. 
\end{proof}
     The NP-hardness of the problem is shown by reducing from a variant of the Knapsack problem. 
	
	
    
	\subsection{When $k$, the number of column players, is a fixed constant.} We show that the problem OPT TRANSITION $(k,m,n)$ is polynomial-time solvable. This is done by reducing the problem of finding an optimal transformation path to the problem of finding the shortest path in a complete directed graph $G(V, E)$.
	
	Given an instance of the problem OPT TRANSITION $(k,m,n)$, the vertex $v \in V$ of graph $G$ is a strategy profile $(r^{v}, \mathcal{C}^{v})$ of the bi-matrix games, where $r^{v} \in \mathscr{R}$ and $\mathcal{C}^{v} = (c_1^{v}, \dots, c_k^{v}), c_i^{v}\in \mathscr{C}$. For a directed edge $e_{vu} \in E$ from vertex $v$ to vertex $u$, its weight $w_{vu}$ is defined to be the reward needed to incentivize the players moving from strategy profile $(r^{v}, \mathcal{C}^{v})$ to $(r^{u}, \mathcal{C}^{u})$, which is $T_{\mathcal{C}^{v}}(r^{u})+T_{r^{v}}(\mathcal{C}^{u})$. 
	
	Let $v^1$ and $v^{*}$ be the vertices corresponding to the initial equilibrium $(r(1),\mathcal{C}(1))$ and the target equilibrium $(r^{*}, \mathcal{C}^{*})$, respectively. The shortest path from $v^{1}$ to $v^{*}$ then corresponds to the optimal transformation path. Since $G$ is a directed graph with non-negative edge weights, and the shortest path problem can be solved in $O(|E|+|V|\log\log|V|)$ time \cite{thorup1999undirected}, we have the following result.
	
	\begin{theorem}\label{thm: constant_row}
		The optimal reward scheme can be computed in time $O(m^2n^{2k})$. That is, the problem OPT TRANSITION $(k,m,n)$ is slicewise polynomial with respect to $k$.
	\end{theorem}
	
	\begin{proof}
		Since each vertex of $G$ corresponds to a strategy profile, the number of vertices $|V| = m\cdot n^k$. Given that $G$ is a complete graph, the number of edges $|E|$ is $\Theta(|V|^2) = \Theta(m^2 n^{2k})$. Thus, the shortest path can be computed in $O(m^2n^{2k})$ time. After getting the shortest path, we then construct the optimal transformation path by setting the strategy profile in round $t$ to be the $(r^{v_t}, \mathcal{C}^{v_t})$, where $v_t$ is the $t$-th vertex on the shortest path. The construction based on the shortest path takes $O(|V|)$ time. As a result, the theorem is proven.
	\end{proof}
	
	
   \subsection{ When $n$, the number of column player strategies, is a fixed constant.} We show that the problem is slicewise polynomial in this case. 
	
	\begin{theorem}\label{thm: const_column}
		When the number of a column player's strategies $n$ is a fixed constant, we can compute the optimal reward scheme in time $O(m^2k^{2n})$.
	\end{theorem}
	\begin{proof}
		Given the fact that the strategy space and utility of the $k$ column players are the same, the $k$ bimatrix games are identical. So, to incentivize the column players' transit to strategy profiles $\mathcal{C}^{j}$ and $\mathcal{C}^{j'}$ when the row player's strategy is $r^i$, respectively, the rewards $T_{r^{i}}(\mathcal{C}^{j})$ and $T_{r^{i}}(\mathcal{C}^{j'})$ are the same, as long as the number of column players choosing each strategy in $\mathscr{C}$ is the same under both $\mathcal{C}^{j}$ and $\mathcal{C}^{j'}$.  
		
		Therefore, we can merge the vertices of the graph $G$ that corresponds to $(r^{i}, \mathcal{C}^{j})$ and $(r^{i}, \mathcal{C}^{j'})$ to form a new vertex. Since there are $k$ column players and $n$ different strategies, there are ${n+k-1 \choose k-1}$ different ways for column players to choose strategies, which means the number of vertices $|V| = m{n+k-1 \choose k-1} = O(mk^{n})$ in the constructed graph. So, the number of edges $|E|$ is $\Theta(|V|^2) = \Theta(m^2 k^{2n})$. 
	\end{proof}


	\section{Approximation Results}\label{Sec: Approx}
	In light of the APX-Hardness of the problem OPT TRANSITION $(k,m,n)$, in this section, we strive to design efficient algorithms that can find a solution within an additive factor of the optimal reward $T$. To this end, we define \emph{alternating path} and use it to design an approximation algorithm.
	
	\subsection{Alternating Path}
	First, we define an alternating path as follows.
	\begin{definition}
		The vertices of an alternating path are either a row player strategy $r(t)$ or a column players' strategy profile $\mathcal{C}(t)$. These vertices alternatingly appear on an alternating path as time epoch $t$ varies. 
	\end{definition}
	The edges of an alternating path are directed and weighted. The weight of edge $( r(t), \mathcal{C}(t+1) )$ is $T_{r(t)}(\mathcal{C}(t+1))$. Namely, the weight is the reward needed to incentivize column players taking up strategy profile $\mathcal{C}(t+1)$ given that the row player's strategy is $r(t)$. Similarly, the weight of edge $( \mathcal{C}(t), r(t+1) )$ is $T_{\mathcal{C}(t)}(r(t+1))$. The cost, $Cost(P)$, of an alternating path $P$ is the sum of all its edges' weights. Then, given a transformation path from the initial equilibrium to the target equilibrium, we can derive two alternating paths with the same length as the transformation path as Figure \ref{fig:alternating-path} demonstrates. Notably, the total cost of these two alternating paths is equal to the cost, $T = \sum_t \{ T_{r(t)}(\mathcal{C}(t+1)) + T_{\mathcal{C}(t)}(r(t+1)) \}$, of the transformation path. As such, we derive the following lemma straightforwardly.

	\begin{figure}[ht]
		\centering
		\includegraphics[width=0.45\textwidth]{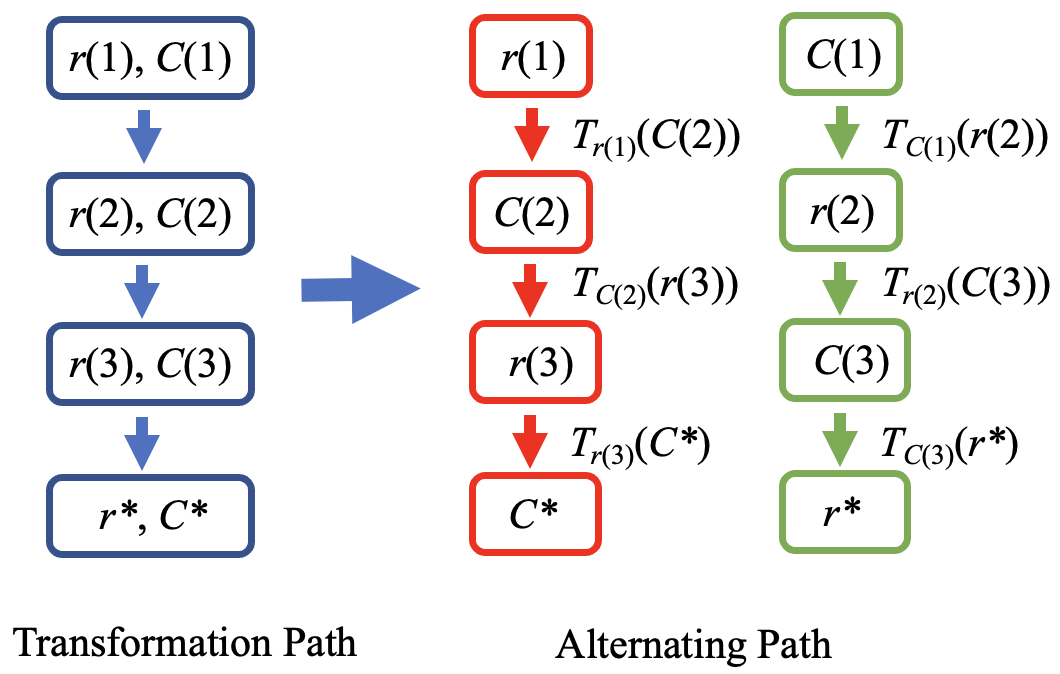}
		\caption{Two alternating paths decomposed from a transformation path.}
		\label{fig:alternating-path}
	\end{figure}

	\begin{lemma}\label{lemma: Lower_bound}
		Given the optimal transformation path from the initial equilibrium $(r(1),\mathcal{C}(1))$ to the target equilibrium $(r^{*}, \mathcal{C}^{*})$, and the two alternating paths decomposed from this transformation path, the cost of the transformation path is at least two times the cost of the alternating path whichever is smaller. 
	\end{lemma}
	\vspace{0pt}
	Now, fixing an initial equilibrium $(r(1),\mathcal{C}(1))$ and a target equilibrium $(r^{*}, \mathcal{C}^{*})$, let us consider all possible alternating paths connecting $r(1)$ or $\mathcal{C}(1)$ and $r^{*}$ or $\mathcal{C}^{*}$, respectively. Without loss of generality, assume the alternating path $r(1) \rightarrow \dots \rightarrow r(l) \rightarrow \mathcal{C}^{*}$ has the smallest cost amongst all these alternating paths. Then, we can use Algorithm~\ref{Alg: Construction} to construct a transformation path from $(r(1),\mathcal{C}(1))$ to $(r^{*}, \mathcal{C}^{*})$ whose cost is twice of the cost of the alternating path $r(1) \rightarrow \dots \rightarrow r(l) \rightarrow \mathcal{C}^{*}$.
	
	\begin{algorithm}[ht]
		\KwInput{An alternating path $P$ with cost $Cost(P)$: $r(1) \rightarrow \dots \rightarrow r(l) \rightarrow \mathcal{C}^{*}$}
		\KwOutput{A transformation path of length $l+1$ with cost $2 \cdot Cost(P)$}
		the 1st vertex $\gets (r(1),\mathcal{C}(1))$\;
		the 2nd vertex $\gets (r(1), \mathcal{C}(2))$\;  
		\While{$3\le t \le l$}
		{
			when $t$ is odd, the $t$-th vertex $\gets (r(t), \mathcal{C}(t-1))$;
			when $t$ is even, the $t$-th vertex $\gets (r(t-1), \mathcal{C}(t))$
		}
		the $(l+1)$-th vertex $\gets (r(l), \mathcal{C}^{*})$\;
		the $(l+2)$-th vertex $\gets (r^{*}, \mathcal{C}^{*})$
		\caption{Transformation Path Construction}\label{Alg: Construction}
	\end{algorithm}

 \setlength{\textfloatsep}{0pt}
	
	In the proof of the following theorem, we will present a constructive procedure for the optimal transformation path. 
	
	\begin{theorem}\label{Theorem:alternating-Path}
		The cost of the transformation path constructed by Algorithm~\ref{Alg: Construction} is two times the cost of the alternating path $r(1) \rightarrow \dots \rightarrow r(l) \rightarrow \mathcal{C}^{*}$. Moreover, the length of the output transformation path is one longer than the input alternating path.
	\end{theorem}
	
	\begin{proof}
		
		We note that the cost of the transformation path constructed by Algorithm~\ref{Alg: Construction} is $$T = \sum_t \{ T_{r(t)}(\mathcal{C}(t+1)) + T_{\mathcal{C}(t)}(r(t+1)) \}.$$   

		Since $(r(1),\mathcal{C}(1))$ and $(r^{*}, \mathcal{C}^{*})$ are two equilibria, the costs $T_{\mathcal{C}(1)}(r(1)) = T_{\mathcal{C}^*}(r^*) = 0$. After expansion, the cost of the transformation path $T = 2 \cdot \{T_{r(1)}(\mathcal{C}(2)) + T_{\mathcal{C}(2)}(r(3)) + \cdots + T_{r(l)}(\mathcal{C}^*)\} = 2 \cdot Cost(P)$. 
		
		As the length of the input alternating path in Algorithm~\ref{Alg: Construction} is $l$, and the output transformation path has a length of $l+1$, the length of the output transformation path is one longer than the input alternating path.
	\end{proof}

	Directly following Lemma~\ref{lemma: Lower_bound} and Theorem~\ref{Theorem:alternating-Path}, we have a corollary as follows.
	
	\begin{corollary}
		The cost of the optimal transformation path from the initial equilibrium $(r(1), \mathcal{C}(1))$ to the target equilibrium $(r^{*}, \mathcal{C}^{*})$ is exactly two times the smallest cost of an alternating path connecting $r(1)$ or $\mathcal{C}(1)$ and $r^{*}$ or $\mathcal{C}^{*}$.
	\end{corollary}
	
	Moreover, if the alternating path with the smallest cost connecting $r(1)$ or $\mathcal{C}(1)$ and $r^{*}$ or $\mathcal{C}^{*}$ is known, then we can construct the optimal transformation path from the initial equilibrium $(r(1),\mathcal{C}(1))$ to the target equilibrium $(r^{*}, \mathcal{C}^{*})$.
	In addition, Algorithm~\ref{Alg: Construction} also allows us to upper bound the length of an optimal transformation path, i.e., the total number of rounds needed to move from the initial equilibrium to the target equilibrium. 
	
	
	
	\begin{theorem}\label{thm:length}
		Given any initial equilibrium $(r(1),\mathcal{C}(1))$ and target equilibrium $(r^{*}, \mathcal{C}^{*})$, there exists an optimal transformation path whose length is at most $2m-1$.
	\end{theorem}
	
	\begin{proof}
		Denote the length of an alternating path with the smallest cost connecting $r(1)$ or $\mathcal{C}(1)$ and $r^{*}$ or $\mathcal{C}^{*}$ by $l$. On one hand, according to Algorithm~\ref{Alg: Construction}, the length of the optimal transformation path constructed on the basis of this alternating path is $l+1$. On the other, we note that there exists such an alternating path on which the row player strategies are not repeated. As otherwise, assume that a row player strategy $r^i$ is repeated on an alternating path as $\dots \rightarrow r^i \rightarrow \dots \rightarrow r^i \rightarrow \mathcal{C}^j \rightarrow \dots$, then we can remove the vertices in between two strategies $r^i$ and modify the alternating path to be $\dots \rightarrow r^i \rightarrow \mathcal{C}^j \rightarrow \dots$. This way, we obtain an alternating path whose cost is no larger than the original alternating path. In addition, as there are at most $m$ row player strategies on this alternating path and there are two edges between two consecutive row player strategies, we have that $l\le 2(m-1)=2m-2$, if this alternating path connects $r(1)$ and $r^{*}$. If this alternating path connects $r(1)$ and $C^{*}$, we further note that $r^{*}$ does not appear in the middle of this alternating path. Otherwise, the cost is not smaller than the cost of an alternating path that connects $r(1)$ and $r^{*}$. Thus, $l\le 1+2(m-1-1)=2m-3$. If this alternating path connects $\mathcal{C}(1)$ and $r^{*}$, we can prove that $r(1)$ does not appear in the middle so that $l \le 1+2(m-1-1)=2m-3$. If this alternating path connects $\mathcal{C}(1)$ and $\mathcal{C}^{*}$, we can prove that either $r(1)$ or $r^{*}$ does not appear in the middle in the same way. Thus, $l\le 2+ 2(m-1-1) = 2m-2$. According to the above analysis of all possible cases, $l\le 2m-2$. Thus, the length of the optimal transformation path is at most $2m-2+1=2m-1$. 
	\end{proof}

	\subsection{Approximation Algorithm}
	In this subsection, we use the properties of the alternating path established above to design an approximation algorithm. We will construct a complete directed graph in which the vertices are row player strategies, and the weight of an edge is the cost of a length-2 alternating path connecting two adjacent vertices. Although the exact minimum cost of these length-2 alternating paths is hard to compute, as otherwise, the problem OPT TRANSITION $(k,m,n)$ will become tractable following our construction, we approximate these costs by rounding the solutions of integer linear programmings. Building upon this complete directed graph, we can assemble an alternating path with nearly-optimal cost between any two row-player strategies. We can extend this approach by an additional handling technique to assemble a minimum-cost alternating path between any row or column player strategies.
	This way, we can approximate the cost of an alternating path connecting $r(1)$ or $\mathcal{C}(1)$ and $r^{*}$ or $\mathcal{C}^{*}$. Building upon the alternating path with the smallest cost among these four paths, we can erect the optimal transformation path between the initial and target equilibria.

	
    We start by constructing a weighted directed complete graph $G(V, E)$. In graph $G$, each vertex $v_i$ corresponds to a strategy $r^{i}$ of the row player. From each vertex $v_{i}$ to a different vertex $v_{j}$, there is an edge $e_{ij}$. The weight $w_{ij}$ of edge $e_{ij}$ is the cost of an alternating path $r^{i} \rightarrow \mathcal{C}_{ij} \rightarrow r^{j}$, where $\mathcal{C}_{ij} \in \mathscr{C}^{k}$ is a column players' strategy profile. Clearly, the smaller $w_{ij}$ is, the better we can approximate the solution to the problem OPT TRANSITION $(k,m,n)$. In Algorithm~\ref{ALG: Alt}, we employ a sub-routine \textbf{WEIGHT($r^{i}$, $r^{j}$)} to compute $w_{ij}$ and the corresponding column players' strategies $\mathcal{C}_{ij}$. Now that we have finished the construction of the graph $G$, we can compute the shortest path between any two vertices $v_i$ and $v_j$. In the meantime, we insert the corresponding column player strategy profiles $\mathcal{C}_{ij}$ into every adjacent vertex $v_i$ and $v_j$, to construct an alternating path (of the equilibrium transition problem) from the shortest path (of graph $G$).
	Clearly, the shortest path $P_{r^{i}r^{j}}$ between $v_i$ and $v_j$ corresponds to an approximately optimal alternating path from $r^{i}$ to $r^{j}$. If both $s_1$ and $s_2$ are row player strategies, then no additional treatment on the shortest path is needed. So Algorithm~\ref{ALG: Alt} returns $P'=P_{s_1s_2}$; otherwise, to obtain an approximately optimal alternating path from $s_1$ to $s_2$, we need to make an additional comparison to append one more column player strategy in front and at the end of the shortest path, so that the cost of the output alternating path is as small as possible.

 \begin{algorithm}[ht!]
		\KwInput{Strategy sets $\mathscr{R}$ and $\mathscr{C}^{k}$; \\
			Payoff matrices $R_{m \times n}$ and $C_{m \times n}$; \\
			Initial equilibrium $(r(1),\mathcal{C}(1))$ and target equilibrium $(r^{*}, \mathcal{C}^{*})$; \\
			$s_1 \in \{ r(1), \mathcal{C}(1) \}$, $s_2 \in \{r^{*}, \mathcal{C}^{*} \}$. }
		\KwOutput{An alternating path from $s_1$ to $s_2$} 
		
		Initialize a complete directed graph $G(V, E)$\;
		\For{$i$ from $1$ to $m$}{
			$v_i \gets r^{i}$\;
			\For{$j$ from $1$ to $m$}{
				$e_{ij}\gets$ an edge from $v_i$ to $v_j$\;
				$(w_{ij}, \mathcal{C}_{ij}) \gets$ WEIGHT($r^{i}, r^{j}$)\;
			}
		}
		\For{$i$ from $1$ to $m$}{
			\For{$j$ from $1$ to $m$}{
				$P_{r^ir^j}\gets$ the shortest path between $v_i$ and $v_j$\;
				\For{each edge $e_{pq}$ along $P_{r^ir^j}$}{
					Insert $C_{pq}$ between $v_p$ and $v_q$\;
				}
			}
		}
		
		\If{$s_1, s_2\in \mathscr{R}$}{
			$P' \gets P_{s_1s_2}$\;
		}
		\If{$s_1\in \mathscr{R}$ and $s_2\in \mathscr{C}^{k}$}{
			$P' \gets \argmin_{P_{s_1q}}\left(Cost(P_{s_1q})+T_{q}(s_2)\right)$\;
			Append $s_2$ to the end of $P'$\;
		}
		\If{$s_1\in \mathscr{C}^{k}$ and $s_2\in \mathscr{R}$}{
			$P' \gets \argmin_{P_{qs_2}}\left(Cost(P_{qs_2})+T_{s_1}(q)\right)$\;
			Insert $s_1$ in the front of $P'$\;
		}
		\If{$s_1, s_2 \in \mathscr{C}^{k}$}{
			$P' \gets \argmin_{P_{pq}}\left(T_{s_1}(p)+Cost(P_{pq})+T_{q}(s_2)\right)$\;
			Append $s_2$ to the end of $P'$\;
			Insert $s_1$ in the front of $P'$\;
		}
		return $P'$
		\caption{Construction of an Approximately Optimal Alternating Path $s_1 \rightarrow \dots \rightarrow s_2$}\label{ALG: Alt}
	\end{algorithm}

\setlength{\textfloatsep}{0pt}
	In the following, we present an integer linear programming formulation used to approximate the cost of a length-2 alternating path in the sub-routine WEIGHT($r^{i}$, $r^{j}$).
	
	\paragraph{Sub-routine WEIGHT($r^{i}$, $r^{j}$).} 
	Denote $x_q$ the number of column players who play the strategy $c^{q} \in \mathscr{C}$. Then the column strategy profile $\mathcal{C}_{ij}\in \mathscr{C}^{k}$ is determined given a tuple $(x_1,\dots, x_n)$.
	The hurdle to approximate $T_{r^i}(\mathcal{C}_{ij}) + T_{\mathcal{C}_{ij}}(r^j)$ is that  $T_{\mathcal{C}_{ij}}(r^{j})$ contains the $\max$ operator. To detour it, we introduce $m$ integer linear programming (ILPs). For the $z$-th ILP, we eliminate the operator $\max$ by introducing a constraint that $r^{z}$ is row player's best response strategy with respect to $\mathcal{C}_{ij}$. By allowing $x_q$ to be a positive number, we further relax these ILPs to linear programmings (LPs). The $z$-th linear programming is shown as follows.
	
	
	\begin{equation}
		\begin{aligned}
			\min \quad & T_{r^{i}}(\mathcal{C}_{ij}) + \sum_{1\le q \le n} x_qR(r^{z}, c^{q}) - \sum_{1\le q \le n} x_qR(r^{j}, c^{q})\\
			\text{s.t.} \quad & 
			\sum_{1\le q \le n} x_qR(r^{z}, c^{q}) \ge \sum_{1\le q \le n}x_qR(r^{z'},c^{q}), \,\,\ \forall z' \in [m]\\
			&x_1+\dots+x_n = k \\
			&x_q \ge 0, \,\,\   \forall q \in [n] 
		\end{aligned}
		\nonumber
	\end{equation}
	
	The objective function is the cost, $T_{r^i}(\mathcal{C}_{ij}) + T_{\mathcal{C}_{ij}}(r^j)$, of the length-2 alternating path $r^i \rightarrow \mathcal{C}_{ij} \rightarrow r^j$. It is linear in the variables $x_q$'s. The first constraint ensures that $r^{z}$ is the row player's best response strategy to column player strategy profile $\mathcal{C}_{ij}$. The remaining constraints state that the total number of column players equals precisely $k$ and $x_q$ is non-negative. After solving these $m$ LPs, we take the one that gives the minimum value $w_{ij}^{*}$ of the objective functions. Without loss of generality, assume that it is the $z$-th LP that achieves the minimum value. Hence, the objective function can be rewritten as $k\max_{p}C(r^{i}, c^{p})-\sum_{1\le q\le n}x_q f(c^{q})$, where $f(c^{q})$ is $C(r^{i}, c^{q})-R(r^{z}, c^{q})+R(r^{j}, c^{q})$. We can retrieve a strategy profile $\mathcal{C}_{ij}$ given the solution $(x_1,\dots,x_n)$. Albeit the solutions $x_q$'s of the LPs are not necessarily integers, for now, we can interpret them as a mixed-strategy profile. 
	
	
	To round the fractional solution to an integral solution, we first round each $x_q$ down to the nearest integer $\lfloor x_q\rfloor$, $q=1,\dots,n$. Let $u= \argmin_{q}\max_{p}R(r^{p}, c^{q})$, we then increase $x_u$ by $k - \sum_{1\le q\le n}\lfloor x_q\rfloor$. Till now, we have derived an integer solution $(\lfloor x_1 \rfloor,\dots, \lfloor x_n \rfloor)$, from which we can retrieve a valid strategy profile $\mathcal{C}_{ij}$. Based on $\mathcal{C}_{ij}$, we obtain an alternating path $r^i \rightarrow \mathcal{C}_{ij} \rightarrow r^j$ with cost $w_{ij}$.
	
	Next, we bound the difference between the optimal objective value of the LPs and the transition cost $w_{ij}$ that we derived from the alternating path. 
	
	\begin{lemma}~\label{Lemma: Two-step}
		$w_{ij} - w_{ij}^{*} \le 2\|R\|_{1,1}+\|C\|_{1,1}$, where $\|R\|_{1,1}$ and $\|C\|_{1,1}$ are the sum of all entries of $R$ and $C$, respectively.
	\end{lemma}

		\begin{proof}
		Denote $(x_1,\dots, x_n)$ and $(x_1', \dots, x_n')$ the fractional solution and the integral solution, respectively. Without loss of generality, assume $z=\arg\max\limits_{1\le p\le m}\sum\limits_{1\le q \le n}x_qR(r^{p}, c^{q})$. We have that
		\begin{equation*}
			\begin{split}
				w_{ij}^{*}&=k\max_{q}C(r^{i}, c^{q})+\sum_{1\le q\le n}x_qR(r^{z},c^{q}) \\
				&-\sum_{1\le q\le n}x_q\left(C(r^{i},c^{q})+R(r^{j},c^{q})\right)
			\end{split}
		\end{equation*}
		In addition, $w_{ij}$ can be written as
		\begin{equation*}
			\begin{split}
				w_{ij}&=k\max_{q}C(r^{i}, c^{q})+\max_{p}\sum_{1\le q\le n}x_q'R(r^{p}, c^{q})\\
				&-\sum_{1\le q\le n}x_q'\left(C(r^{i},c^{q})+R(r^{j},c^{q})\right)
			\end{split}
		\end{equation*}
		
		The difference between $w_{ij}$ and the cost of $(x_1,\dots, x_n)$ can be computed directly as 
		\begin{equation*}
			\begin{split}
				\max_{p} &\sum_{1\le q\le n}x_q'R(r^{p}, c^{q})-\sum_{1\le q\le n}x_qR(r^{z}, c^{q})\\
				&+\sum_{1\le q\le n}(x_q-x_q')\left(C(r^{i},c^{q})+R(r^{j},c^{q})\right)
			\end{split}
		\end{equation*}
		For the term $\sum_{1\le q\le n}(x_q-x_q')\left(C(r^{i},c^{q})+R(r^{j},c^{q})\right)$, according to the rounding procedure, $x_q-x_q'< 1$, this term is no more than $\|R\|_{1,1}+\|C\|_{1,1}$. Let $u= \argmin_{q}\max_{p}R(r^{p}, c^{q})$. Since $x_q'\le x_q$ for all $q\not=u$ and $x_u'-n\le x_u$, $\max_{p}\sum_{1\le q\le n}x_q'R(r^{p}, c^{q})-\sum_{1\le q\le n}x_qR(r^{s}, c^{q})$ is no more than $n\max_{p}R(r^{p}, c^{u})$, which is no more than $\|R\|_{1,1}$. Combining the analysis above, we prove the lemma.
	\end{proof}

	\begin{theorem}\label{Theorem: App_Alt}
		Algorithm~\ref{ALG: Alt} returns an alternating path between $s_1$ and $s_2$ whose cost is at most $m(2\|R\|_{1,1}+\|C\|_{1,1})$ more than the cost of the optimal alternating path between $s_1$ and $s_2$.
	\end{theorem}

	\begin{proof}
		If $s_1, s_2\in \mathscr{R}$, we first quantify the cost of the optimal alternating path. Assume that there exists a directed complete graph $G^{*}$ of which the vertex $v_i^{*}$ is $r^{i}$ and the weight of a directed edge $e_{ij}^{*}$ from $v_{i}^{*}$ to $v_{j}^{*}$ is $w_{ij}^{*}$. Recall that $w_{ij}^{*}$ is the minimum value of the objective functions via solving $m$ LPs. Thus, the cost of any two-round alternating path between $r^{i}$ and $r^{j}$ is at least $w_{ij}^{*}$. Let $P_{s_1s_2}^{*}$ denote the shortest path from $s_1$ to $s_2$ in the graph $G^{*}$. The cost of the optimal alternating path is $\sum_{e_{ij}^{*}\in P_{s_1s_2}^{*}}w_{ij}^{*}$. Clearly, we can find the same path from $s_1$ to $s_2$ in the graph $G$ with a cost of $\sum_{e_{ij}^{*}\in P_{s_1s_2}^{*}}w_{ij}$. According to Lemma~\ref{Lemma: Two-step} and the fact that $P_{s_1s_2}^{*}$ contains no more than $m$ edges, we have
		$$\sum_{e_{ij}^{*}\in P_{s_1s_2}^{*}}(w_{ij}-w_{ij}^{*})\le m(2\|R\|_{1,1}+\|C\|_{1,1}).$$
		As Algorithm~\ref{ALG: Alt} finds the shortest path between $s_1$ and $s_2$ in graph $G$, the additive error is upper bounded by $m(2\|R\|_{1,1}+\|C\|_{1,1})$ as well.
		
		If $s_1\in \mathscr{R}$ and $s_2\in \mathscr{C}^{k}$, without loss of generality, assume the second last vertex of the optimal alternating path between $s_1$ and $s_2$ is $r^{q}\in \mathscr{R}$. Let $P_{s_1r^{q}}$ denote the shortest path from $s_1$ to $r^{q}$ in graph $G^{*}$. The cost of the optimal path is $\sum_{e_{ij}^{*}\in P_{s_1r^{q}}^{*}}w_{ij}^{*}+T_{r^{q}}(s_2)$. According to Lemma~\ref{Lemma: Two-step} and the fact that $P_{s_1s_2}^{*}$ contains no more than $m$ edges, we also have
		$$\sum_{e_{ij}^{*}\in P_{s_1r^{q}}^{*}}(w_{ij}-w_{ij}^{*})\le m(2\|R\|_{1,1}+\|C\|_{1,1})$$
		As Algorithm~\ref{ALG: Alt} returns a path with a cost no more than $\sum_{e_{ij}^{*}\in P_{s_1r^{q}}^{*}}w_{ij}+T_{r^{q}}(s_2)$, the additive error is upper bounded by $m(2\|R\|_{1,1}+\|C\|_{1,1})$ as well. For the other two cases, the additive error is upper bounded in the same way.
	\end{proof}

	Finally, we bound the additive approximation error. 
	
	\begin{theorem}\label{Theorem: Approx}
		Denote $\OPT$ the cost of the optimal transformation path from $(r(1),\mathcal{C}(1))$ to $(r^{*}, \mathcal{C}^{*})$. We can implement Algorithm~\ref{ALG: Alt} to find an approximately optimal alternating path and then Algorithm~\ref{Alg: Construction} to construct a transformation path with a cost $\OPT + 2m(2\|R\|_{1,1}+\|C\|_{1,1})$ in polynomial time.
	\end{theorem}

	\begin{proof}
		Through Algorithm~\ref{ALG: Alt}, we obtain four alternating paths that start from either $r(1)$ or $\mathcal{C}(1)$ and ends at either $r^{*}$ or $\mathcal{C}^{*}$. Taking the one with the minimum cost as the input of Algorithm~\ref{Alg: Construction}, we obtain a transformation path. According to Lemma~\ref{lemma: Lower_bound} and Theorem~\ref{Theorem: App_Alt}, the difference between the cost of this transformation path and the optimal transformation path is upper bounded by $2m(2\|R\|_{1,1}+\|C\|_{1,1})$.
	\end{proof}
	
	We note that the advantage of this additive approximation is that it is independent of the number of column players $k$. In many practical scenarios such as increasing the uptake of electric vehicles and retail store businesses, $k$ is the number of drivers/customers which is significantly larger than the number of player strategies. 
	
	\section{Tractable Cases}
	In this section, we turn our attention to the cases in which the optimal transformation path can be found in polynomial time. Recall the example in the Introduction: in the game, the row player is a service provider and the column players are the customers. We assume that all possible locations are distributed on a straight line. That is, the players' strategy space $\mathscr{R}=\mathscr{C}$. In particular, the row player is limited to choosing a location which is one of the column player strategies. The service provider seeks to minimize the sum of the Euclidean distance between its location and the locations of customers. A customer's payoff is negatively correlated with its distance to the service provider. With a slight abuse of notations, we denote $r^i,c^j$ the axis of the locations on the line. W.l.o.g., we assume they are both sorted in increasing orders.
	
	\begin{definition}
		In this one-dimensional domain, a player's payoff is \textit{single-peaked} if it is maximized at a single point on the line, and the further its location to this point, the less its payoff. A player's payoff is \textit{exact-distance} if it is equal to the negative value of the difference between its location to a single point on the line. 
	\end{definition}
	
	In particular, we are interested in the case in which the row player's payoff is exact-distance and the column players' payoff is single-peaked. That is, given the strategy profile $(r^{i}, c_{1}^{j_1},\dots,c_{k}^{j_k})$, $j_l \in [m]$, $l=1,\dots,k$, the row player's payoff is $\sum_{l=1}^{k} R(r^{i}, c_{l}^{j_l}) = - \sum_{l=1}^{k} |r^{i} - c_{l}^{j_l}|$, and the column player $l$'s payoff is $C(r^{i}, c_{l}^{j_l}) = g(|r^{i} - c_{l}^{j_l}|)$, where $g(\cdot)$ is monotone decreasing and $g(0)=0$. In this case, we show that the problem is tractable.
	
	\begin{theorem}\label{thm:single-peak}
		When the row player's payoff is exact-distance and the column players' payoff is single-peaked, the optimal transformation path from an initial equilibrium to a target equilibrium can be found in polynomial time.
	\end{theorem}

	\begin{proof}
		In Theorem \ref{Theorem: App_Alt}, we have shown that if $WEIGHT(r^i,r^j)$ finds a length-2 optimal alternating path from $r^i$ to $r^j$, then Algorithm~\ref{ALG: Alt} returns the optimal transformation path. 
		Thus, we only need to devise an algorithm for finding the column strategy profile $\mathcal{C}=(c^1,\dots,c^k)$ for the alternating path $r^{i} \rightarrow \mathcal{C} \rightarrow r^{j}$. 
		
		W.l.o.g., we assume the two row strategy locations $r^{i} < r^{j}$, and the $k$ column strategy locations are sorted in order of $c^1\le c^2\le \dots \le c^k$. Then, the cost of the alternating path $r^{i} \rightarrow \mathcal{C} \rightarrow r^{j}$ is
		\begin{equation*}
			\begin{split}
				&T_{r^{i}}(\mathcal{C})+T_{\mathcal{C}}(r^{j}) \\
				=& \Big\{ k\cdot \max_{c^* \in \mathscr{C}}g(|r^i-c^*|) - \sum_{l=1}^k g(|r^i-c^l|) \Big\} \\
				&	+ \Big\{ \max_{r^* \in  \mathscr{R}} \Big[ - \sum_{l=1}^k |r^*-c^l| \Big] - \Big[ - \sum_{l=1}^k |r^j-c^l| \Big] \Big\}.
			\end{split}
		\end{equation*}
		Note that $\mathscr{R}=\mathscr{C}$ since both the row player strategy and column player strategies are locations on a line, we have that $\max\limits_{c^* \in \mathscr{C}}g(|r^i-c^*|) = 0$ by setting $c^* = r^i$. 
		Since $c^l, l=1,\dots,k$ is sorted in an ascending order, the sum $\sum_{l=1}^k |r^*-c^l|$ achieves the minimum when $r^*=c_h$, where $h=\lfloor \frac{k+1}{2}\rfloor$.
		Therefore, the cost of the alternating path $r^{i} \rightarrow \mathcal{C} \rightarrow r^{j}$ becomes 
		\begin{equation*}
			\begin{split}
				- &\sum_{l=1}^k g(|r^i-c^l|) - \sum_{l\le h}(c_h-c^l) + \\
				& \big(-\sum_{l>h}(c^l-c_h) + \sum_{l=1}^k |r^j - c^l|\big) \label{eq1}
			\end{split}
		\end{equation*}
		
		We note that in the minimum-cost alternating path $r^{i} \rightarrow \mathcal{C} \rightarrow r^{j}$, the $c^l$'s of $\mathcal{C}$ locate in between $r^i$ and $r^j$. That is, $r^i \le c^l \le r^j, l=1,\dots,k$. We will show that $r^i \le c^l , l=1,\dots,k$, and the proof for the other direction is similar.
		
		Assume for contradiction that in $\mathcal{C} = (c^{1},\dots,c^{k})$, there exists $d \in \{1,\dots,k\}$ such that $c^1 \le \cdots \le c^d < r^i \le c^{d+1}$. Denote a different column player strategy profile $\mathcal{C}' = (c'^{1},\dots,c'^{k}) = (r^i,\dots,r^i,c^{d+1},\dots,c^{k})$. 
		That is, the first $d$ strategies are moved to $r^i$. We will show that $\mathcal{C}'$ yields a lower transition cost than $\mathcal{C}$ as the intermediate step of a length-2 alternating path.  
		
		There are two subcases here, either $d\le h$ or $d> h$. We consider the first case, and the analysis for the second case is similar.
		
		As the order of individual column player's strategy location is unchanged, i.e., $c'^{1} \le \dots \le c'^{k}$, in $T_{r^{i}}(\mathcal{C}')+T_{\mathcal{C}'}(r^{j})$, it still holds that 
		the sum $\sum_l|r^*-c'^l|$ achieves its minimum when $r^*=c_h$, where $h=\lfloor \frac{k+1}{2}\rfloor$. Therefore, the cost of the alternating path $r^{i} \rightarrow \mathcal{C}' \rightarrow r^{j}$ is 
		\begin{equation}\label{eq2}
			\begin{split}
				& - \sum_{l=1}^k g(|r^i-c'^l|) - \sum_{l\le h}(c_h-c'^l) + \\
				&\big(-\sum_{l>h}(c'^l-c_h) + \sum_{l=1}^k |r^j - c'^l|\big) \nonumber\\
				= &-\sum_{l>d}g(c'^l-r^i)+d\cdot(r^j-r^i)+\sum_{l>d}|r^j-c^l| \\
				&	-\sum_{l\le d}(c_h-r^i)-\sum_{d< l \le h}(c_h-c^l)-\sum_{l>h}(c^l-c_h)
			\end{split}
		\end{equation}
		
		To compare the transition costs generated by $\mathcal{C}$ and $\mathcal{C}'$, we deduct Equation \eqref{eq2} from Equation \eqref{eq1} and get that
		\begin{equation*}
			\begin{split}
				&-\sum_{l\le d}g(|r^i-c^l|)
				+\sum_{l\le d}(r^j-c^l) \\
				&+ \sum_{l\le d}(c_h-r^i)+ 
				\big(- d(r^j-r^i) - \sum_{l\le d}(c_h-c^l)\big)\nonumber \\
				=&-\sum_{l\le d}g(|r^i-c^l|)+\sum_{l\le d}(r^i-c^l) - \sum_{l\le d}(r^i-c^l) \\
				=&-\sum_{l\le d}g(|r^i-c^l|)>0
			\end{split}
		\end{equation*}
		Thus, strategy profile $\mathcal{C}'$ yields a lower transition cost.
		
		Therefore, in the minimum-cost alternating path $r^{i} \rightarrow \mathcal{C} \rightarrow r^{j}$, the $c^l$'s of $\mathcal{C}$ locate in between $r^i$ and $r^j$. This fact enables us to get rid of the absolute value operator in Equation \eqref{eq1}. So, the transition cost of the alternating path $r^{i} \rightarrow \mathcal{C} \rightarrow r^{j}$ is
		\begin{equation*}
			\begin{split}
				- &\sum_{l=1}^k g(c^l-r^i) - \sum_{l\le h}(c_h-c^l)-\sum_{l>h}(c^l-c_h) + \sum_{l=1}^k (r^j - c^l)\\
				=& \,\ \sum_{l\le h}\big[-g(c^l-r^i)+(r^j-c_h)\big] \\
				&+\sum_{l>h}\big[-g(c^l-r^i)+(r^j+c_h-2c^l)\big].
			\end{split}
		\end{equation*}
		Since $g(\cdot)$ is monotone decreasing and $g(0)=0$, to minimize the transition cost, we should set $c^l=r^i$ for $l\le h$. Then, the cost becomes
		\begin{align}\label{eq3}
			h(r^j-c_h) + \sum\limits_{l>h}\big[-g(c^l-r^i)+(r^j+c_h-2c^l)\big].
		\end{align}
		Next, for each possible $c_h \in \{c^1,\dots,c^n\} \cap [r^i,r^j]$, denote strategy $c^{l^*}=\argmin_{c^l\ge c_h} \{-g(c^l-r^i)+c_h-2c^l\}$. By comparing the value of Equation \eqref{eq3} for each combination of $c_h$ and $c^{l^*}$, we can obtain the optimal pair of these strategies that minimize the transition cost. Denote them by $c^{h^*}$ and $c^{l^*}$. Since there are at most $n$ column strategies, there are at most $n^2$ combinations that can be efficiently enumerated.  
		
		Till now, we have computed the strategy profile $\mathcal{C}=(c^1,\dots,c^k)=(r^i,\dots,r^i,c^{h^*},c^{l^*},\dots,c^{l^*})$ that minimizes the transition cost of the alternating path $r^{i} \rightarrow \mathcal{C} \rightarrow r^{j}$.
		
		Finally, we adopt Algorithm \ref{ALG: Alt} to find the optimal alternating path by implementing the $WEIGHT(r^i,r^j)$ procedure. Then, by implementing Algorithm~\ref{Alg: Construction}, we can construct the optimal transformation path from an initial equilibrium to a target equilibrium.
	\end{proof}
	In this setting, the payoff matrices have nice properties which enable us to find the optimal alternating paths. 
 

\section{Conclusions and Future Work}
	In this paper, we formulated an optimization problem to find the optimal transformation path from an initial equilibrium to a more desirable one. In this game, players move simultaneously in each round. We presented comprehensive complexity analyses of the problem. We proved that the problem is APX-hard when the number of the row player strategies, the number of the column player strategies, and $k$ are input sizes. Furthermore, we showed that the problem is slicewise polynomial with respect to $k$ and $n$, respectively. As for the parameter $m$, we proved that the problem is NP-Hard even when $m=2$. Besides the hardness result, we designed a polynomial-time approximation algorithm with bounded additive error. Moreover, the approximation error is independent of the number of column players $k$. Finally, we considered cases where we can find the optimal transformation path in polynomial time.
	
    It is important to acknowledge that the problem we have analyzed presents intractability in various setups, making any generalizations a formidable task. Viable solutions would likely require the application of additional restrictions or careful consideration of special cases. However, the realm of possibilities for further exploration is vast and captivating. Firstly, delving into the extension of our study to encompass scenarios where column players possess diverse strategies and payoff matrices holds great interest. Secondly, exploring the implications of incorporating randomized strategies and how they might alter the overall dynamics of the game poses an intriguing challenge. Thirdly, developing efficient algorithms for other cases, such as utilizing alternative distance measures to define player strategies corresponding to locations on general graphs and their associated payoffs, presents an enticing avenue for investigation. Additionally, broadening the scope of the study to incorporate multiple row players opens up exciting possibilities for research and discovery.

\section*{Acknowledgments}
Zihe Wang was partially supported by the National Natural
Science Foundation of China (Grant No. 62172422). Jie Zhang was partially supported by a Leverhulme Trust Research Project Grant (2021 -- 2024) and the EPSRC grant (EP/W014912/1).

\bibliography{aaai24}

\begin{thebibliography}{19}
\providecommand{\natexlab}[1]{#1}

\bibitem[{Amiet et~al.(2021)Amiet, Collevecchio, Scarsini, and
  Zhong}]{DBLP:journals/mor/AmietCSZ21}
Amiet, B.; Collevecchio, A.; Scarsini, M.; and Zhong, Z. 2021.
\newblock Pure Nash Equilibria and Best-Response Dynamics in Random Games.
\newblock \emph{Math. Oper. Res.}, 46(4): 1552--1572.

\bibitem[{Anshelevich et~al.(2008)Anshelevich, Dasgupta, Kleinberg, Tardos,
  Wexler, and Roughgarden}]{DBLP:journals/siamcomp/AnshelevichDKTWR08}
Anshelevich, E.; Dasgupta, A.; Kleinberg, J.~M.; Tardos, {\'{E}}.; Wexler, T.;
  and Roughgarden, T. 2008.
\newblock The Price of Stability for Network Design with Fair Cost Allocation.
\newblock \emph{{SIAM} J. Comput.}, 38(4): 1602--1623.

\bibitem[{Augustine et~al.(2015)Augustine, Caragiannis, Fanelli, and
  Kalaitzis}]{DBLP:journals/algorithmica/AugustineCFK15}
Augustine, J.; Caragiannis, I.; Fanelli, A.; and Kalaitzis, C. 2015.
\newblock Enforcing Efficient Equilibria in Network Design Games via Subsidies.
\newblock \emph{Algorithmica}, 72(1): 44--82.

\bibitem[{Deng and Conitzer(2017)}]{DBLP:conf/aaai/DengC17}
Deng, Y.; and Conitzer, V. 2017.
\newblock Disarmament Games.
\newblock In Singh, S.~P.; and Markovitch, S., eds., \emph{Proceedings of the
  Thirty-First {AAAI} Conference on Artificial Intelligence, February 4-9,
  2017, San Francisco, California, {USA}}, 473--479. {AAAI} Press.

\bibitem[{Deng and Conitzer(2018)}]{DBLP:conf/aaai/DengC18}
Deng, Y.; and Conitzer, V. 2018.
\newblock Disarmament Games With Resource.
\newblock In \emph{Proceedings of the Thirty-Second {AAAI} Conference on
  Artificial Intelligence, (AAAI-18), the 30th innovative Applications of
  Artificial Intelligence (IAAI-18), and the 8th {AAAI} Symposium on
  Educational Advances in Artificial Intelligence (EAAI-18)}, 981--988. {AAAI}
  Press.

\bibitem[{Deng, Tang, and Zheng(2016)}]{DBLP:conf/atal/DengTZ16}
Deng, Y.; Tang, P.; and Zheng, S. 2016.
\newblock Complexity and Algorithms of K-implementation.
\newblock In \emph{Proceedings of the 2016 International Conference on
  Autonomous Agents {\&} Multiagent Systems, Singapore, May 9-13, 2016}, 5--13.
  {ACM}.

\bibitem[{Eidenbenz et~al.(2007)Eidenbenz, Oswald, Schmid, and
  Wattenhofer}]{DBLP:conf/cocoa/EidenbenzOSW07}
Eidenbenz, R.; Oswald, Y.~A.; Schmid, S.; and Wattenhofer, R. 2007.
\newblock Mechanism Design by Creditability.
\newblock In \emph{Combinatorial Optimization and Applications, First
  International Conference, {COCOA} 2007}, volume 4616 of \emph{Lecture Notes
  in Computer Science}, 208--219. Springer.

\bibitem[{Feldman, Snappir, and Tamir(2017)}]{DBLP:conf/sagt/FeldmanST17}
Feldman, M.; Snappir, Y.; and Tamir, T. 2017.
\newblock The Efficiency of Best-Response Dynamics.
\newblock In Bil{\`{o}}, V.; and Flammini, M., eds., \emph{Algorithmic Game
  Theory - 10th International Symposium, {SAGT} 2017, L'Aquila, Italy,
  September 12-14, 2017, Proceedings}, volume 10504 of \emph{Lecture Notes in
  Computer Science}, 186--198. Springer.

\bibitem[{Heinrich et~al.(2022)Heinrich, Jang, Mungo, Pangallo, Scott, Tarbush,
  and Wiese}]{heinrich2022bestresponse}
Heinrich, T.; Jang, Y.; Mungo, L.; Pangallo, M.; Scott, A.; Tarbush, B.; and
  Wiese, S. 2022.
\newblock Best-response dynamics, playing sequences, and convergence to
  equilibrium in random games.
\newblock arXiv:2101.04222.

\bibitem[{Hopkins(1999)}]{hopkins1999note}
Hopkins, E. 1999.
\newblock A note on best response dynamics.
\newblock \emph{Games and Economic Behavior}, 29(1-2): 138--150.

\bibitem[{Karp(1972)}]{DBLP:conf/coco/Karp72}
Karp, R.~M. 1972.
\newblock Reducibility Among Combinatorial Problems.
\newblock In Miller, R.~E.; and Thatcher, J.~W., eds., \emph{Proceedings of a
  symposium on the Complexity of Computer Computations, held March 20-22, 1972,
  at the {IBM} Thomas J. Watson Research Center, Yorktown Heights, New York,
  {USA}}, The {IBM} Research Symposia Series, 85--103. Plenum Press, New York.

\bibitem[{Koutsoupias and
  Papadimitriou(2009)}]{DBLP:journals/csr/KoutsoupiasP09}
Koutsoupias, E.; and Papadimitriou, C.~H. 2009.
\newblock Worst-case equilibria.
\newblock \emph{Comput. Sci. Rev.}, 3(2): 65--69.

\bibitem[{Leslie, Perkins, and Xu(2020)}]{leslie2020best}
Leslie, D.~S.; Perkins, S.; and Xu, Z. 2020.
\newblock Best-response dynamics in zero-sum stochastic games.
\newblock \emph{Journal of Economic Theory}, 189: 105095.

\bibitem[{Monderer and Tennenholtz(2004)}]{DBLP:journals/jair/MondererT04}
Monderer, D.; and Tennenholtz, M. 2004.
\newblock K-Implementation.
\newblock \emph{J. Artif. Intell. Res.}, 21: 37--62.

\bibitem[{Monderer and Tennenholtz(2009)}]{DBLP:journals/ai/MondererT09}
Monderer, D.; and Tennenholtz, M. 2009.
\newblock Strong mediated equilibrium.
\newblock \emph{Artif. Intell.}, 173(1): 180--195.

\bibitem[{Nisan and Ronen(1999)}]{DBLP:conf/stoc/NisanR99}
Nisan, N.; and Ronen, A. 1999.
\newblock Algorithmic Mechanism Design (Extended Abstract).
\newblock In Vitter, J.~S.; Larmore, L.~L.; and Leighton, F.~T., eds.,
  \emph{Proceedings of the Thirty-First Annual {ACM} Symposium on Theory of
  Computing, May 1-4, 1999, Atlanta, Georgia, {USA}}, 129--140. {ACM}.

\bibitem[{Papadimitriou(2011)}]{DBLP:conf/www/Papadimitriou11}
Papadimitriou, C.~H. 2011.
\newblock Games, algorithms, and the Internet.
\newblock In Srinivasan, S.; Ramamritham, K.; Kumar, A.; Ravindra, M.~P.;
  Bertino, E.; and Kumar, R., eds., \emph{Proceedings of the 20th International
  Conference on World Wide Web, {WWW} 2011, Hyderabad, India, March 28 - April
  1, 2011}, 5--6. {ACM}.

\bibitem[{Rozenfeld and Tennenholtz(2007)}]{DBLP:conf/ijcai/RozenfeldT07}
Rozenfeld, O.; and Tennenholtz, M. 2007.
\newblock Routing Mediators.
\newblock In \emph{{IJCAI} 2007, Proceedings of the 20th International Joint
  Conference on Artificial Intelligence}, 1488--1493.

\bibitem[{Thorup(1999)}]{thorup1999undirected}
Thorup, M. 1999.
\newblock Undirected single-source shortest paths with positive integer weights
  in linear time.
\newblock \emph{Journal of the ACM (JACM)}, 46(3): 362--394.

\end{thebibliography}

\end{document}